\documentclass[copyright]{eptcs}

\providecommand{\event}{QAPL 2013} 
\usepackage{amssymb,amsmath,amsfonts,amsthm}
\usepackage{bm}
\usepackage{mathtools}
\usepackage{stmaryrd}
\usepackage[basic]{complexity}
\usepackage{ifthen}
\usepackage{hyperref}

\newboolean{lipicsstyle}
\setboolean{lipicsstyle}{false}
\newboolean{eptcsstyle}
\setboolean{eptcsstyle}{false}

\setboolean{eptcsstyle}{true}

\theoremstyle{plain}
\iflipicsstyle
\else 
\newtheorem{theorem}{Theorem}[section]
\newtheorem{corollary}[theorem]{Corollary}
\newtheorem{definition}[theorem]{Definition}
\newtheorem{lemma}[theorem]{Lemma}
\newtheorem*{remark}{Remark}
\fi
\newtheorem{proposition}[theorem]{Proposition}
\newtheorem{fact}[theorem]{Fact}

\def\ldouble{\langle\!\!\langle}
\def\rdouble{\rangle\!\!\rangle}

\newcommand{\subword}{\sqsubseteq}
\newcommand{\supword}{\sqsupseteq}
\newcommand{\Cu}{C_{\uparrow}}

\newcommand{\Cd}{C_{\downarrow}}

\newcommand{\Ku}{K_{\uparrow}}
\newcommand{\Kd}{K_{\downarrow}}
\newcommand\Pre{{\mathit{Pre}}}

\newcommand\Post{{\mathit{Post}}}
\newcommand\wPre{{\widetilde{\Pre}}}

\newcommand\Dist{{\mathit{Dist}}}

\newcommand\bfF{{\mathbf{F}}}

\newcommand{\Nat}{{\mathbb{N}}} 

\renewcommand{\epsilon}{\varepsilon} 
\renewcommand{\setminus}{\smallsetminus} 
\renewcommand{\phi}{\varphi} 
\newcommand{\sat}{\models}      

\newcommand{\Reg}{{\mathbf{Reg}}}

\newcommand{\Moves}{\mathsf{Moves}}
\newcommand{\Ena}{\mathsf{Enabled}}

\newcommand{\Ch}{{\mathsf{C}}} 
\newcommand{\ttc}{{\mathtt{ch}}}
\newcommand{\PR}{{{\mathbb{P}}}} 
\newcommand{\bfP}{{\textbf{P}}} 
\newcommand{\Rcal}{{\mathcal{R}}}   

\newcommand{\Tcal}{{\mathcal{T}}}
\newcommand{\Next}{{\bigcirc}}
\newcommand{\Conf}{{\mathit{Conf}}} 

\newcommand{\ConfA}{{\Conf_{\!\!A}}}
\newcommand{\ConfB}{{\Conf_{\!\!B}}}
\newcommand{\ConfS}{{\Conf_{\!\!S}}}
\newcommand{\WinA}{{\ldouble A \rdouble}} 

\newcommand{\HA}{{H_{1,r}}}
\newcommand{\Pexists}{{\Pre^{\exists}}}
\newcommand{\Pforall}{{\Pre^{\forall}}}
\newcommand{\PxA}{{\Pre^{\otimes}_{\!A}}}

\newcommand{\env}{{\textit{env}}}

\newcommand{\Mess}{{\mathsf{M}}} 
\newcommand{\op}{{\mathit{op}}}
\newcommand{\Ord}{{\mathit{Ord}}}

\newcommand{\sizeof}[1]{{\mathopen{|}#1\mathclose{|}}}

\newcommand{\overto}[1]{\xrightarrow{\!\!#1\!\!}}
\newcommand{\step}[1]{\overto{#1}} 

\def\leadsto{\rightsquigarrow}
\newcommand{\dup}{\mathrm{dup}}
\newcommand{\Tdup}{\Tcal_\dup}
\newcommand{\invTdup}{\Tcal^{-1}_\dup}
\newcommand{\winvTdup}{\widetilde{\Tcal^{-1}_\dup}}

\newcommand{\egdef}{\stackrel{\mbox{\begin{scriptsize}def\end{scriptsize}}}{=}}
\newcommand{\equivdef}{\stackrel{\mbox{\begin{scriptsize}def\end{scriptsize}}}{\Leftrightarrow}}

\newcommand{\G}{{\mathcal{G}}} 

\begin{document}

        \title{Solving Stochastic B\"uchi Games on Infinite Arenas
        with a Finite Attractor\thanks{Supported by
Grant ANR-11 BS02-001.}}

\iflipicsstyle

        \author[1]{Nathalie Bertrand}
        \author[2]{Philippe Schnoebelen}
        \authorrunning{Bertrand \& Schnoebelen}
        \affil[1]{Inria Rennes Bretagne Atlantique}
        \affil[2]{LSV, ENS Cachan, CNRS}

\else
\ifthenelse{\boolean{eptcsstyle}}{ 
        \author{
Nathalie Bertrand 
\institute{Inria Rennes Bretagne Atlantique}
\and
Philippe Schnoebelen
\institute{LSV, ENS Cachan, CNRS}
        }
\def\titlerunning{Solving Stochastic B\"uchi Games on Infinite Arenas with a Finite Attractor}
\def\authorrunning{N.~Bertrand and Ph.~Schnoebelen}
}{                              

        \author{Nathalie Bertrand \and Philippe Schnoebelen}

}
\fi

\maketitle

\begin{abstract}
  We consider games played on an infinite probabilistic arena where
  the first player aims at satisfying generalized B{\"u}chi objectives
  almost surely, i.e., with probability one. We provide a fixpoint
  characterization of the winning sets and associated winning
  strategies in the case where the arena satisfies the finite-attractor
  property. From this we directly deduce the decidability of these
  games on probabilistic lossy channel systems.
\end{abstract}


\section{Introduction}
\label{sec-intro}

2-player stochastic games are games where two players, Alice and
Bob, interact in a probabilistic environment. Given an objective
formalized, e.g., as an $\omega$-regular condition, the goal for
Alice is to maximize the probability to fulfill the condition, against
any behaviour of her opponent. Qualitative questions ask whether Alice
can win almost-surely (resp.\ positively) from a given initial
configuration. Solving a stochastic game then amounts to deciding the
latter question, as well as providing winning strategies for the
players.  In the case where the arena is finite, the literature offers
several general results on the existence of optimal strategies, the
determinacy of the games, and algorithmic methods for computing
solutions, when the objectives range in complexity from simple
reachability objectives to arbitrary Borel
objectives~\cite{shapley53,condon-ic92,CJH-csl03}.

For infinite arenas, general results are scarce and mostly concern
purely mathematical, non-algo\-rith\-mi\-cal, aspects, such as
determinacy~\cite{martin98,brazdil2013}. An obvious explanation for the lack of
algorithmical results is that infinite-state spaces usually lead to
undecidable, even highly undecidable, problems. This already happens
for the simplest objectives with a single player and no stochastic aspects.

Decidability can be regained for infinite arenas if it is known that they
are generated in some specific way.
The stochastic games on infinite arenas considered in the field of algorithmic
verification  originate from
classical ---i.e., non-stochastic and non-competitive--- infinite-state models. Prominent examples with positive
results are stochastic games on systems with
recursion~\cite{EY-lmcs08,BBKO-ic11},  on one-counter
automata~\cite{BBE-fsttcs10,BBEK-icalp11}, and on
lossy channel systems~\cite{BBS-acmtocl06,ABdAMS-fossacs08}. In all
these examples, the description of winning sets and winning strategies
is specific to the underlying infinite-state model, and rely
on ad-hoc techniques.

In this paper we follow a more generic approach, and study stochastic
games on (finite-choice) infinite arenas where we only assume the
finite-attractor property~\cite{BBS-ipl}. That is, we assume that some
finite set of configurations is almost-surely visited infinitely
often, independently of the behaviors of the players.

\paragraph{Our contributions.}
Our first contribution is a simple fixpoint characterization of the winning
sets and associated winning strategies for generalized B\"uchi
objectives with probability one. The characterization is not concerned with
computability and applies to any finite-choice countable arena with a finite attractor.
We use $\mu$-calculus notation to define, and reason about, the winning
sets and winning strategies: one of our goals is to give a \emph{fully
  detailed} generic correctness proof (we used the characterization 
without proof in~\cite[section 9.2]{BS-fmsd2012}).

Our second contribution is an application of the above
characterization to prove the computability of winning sets (for
generalized B\"uchi objectives) in arenas generated by probabilistic
lossy channel systems (PLCSs).  Rather than using ad-hoc reasoning, we
follow the approach advocated in~\cite{BBS-lpar,BS-fmsd2012} and use
a generic finite-time convergence theorem for well-structured
transition systems (more generally: for fixpoints over the powersets
of WQO's). This allows us to infer the computability and the regularity of
the winning sets directly from the fact that their fixpoint
characterization uses ``upward-guarded'' fixpoint terms built on
regularity-preserving operators. The method easily
accommodates arbitrary regular arena partitions,
PLCSs extended with regular guards, and other kinds of unreliability.

\paragraph{Related work on lossy channel systems.}
An early positive result for stochastic games on probabilistic lossy
channel systems is the decidability of single-player reachability or
B\"uchi games with probability one (dually, safety or co-B\"uchi with
positive probability)~\cite{BBS-acmtocl06}. Abdulla \emph{et al.} then
proved the determinacy and decidability of two-player stochastic games
on PLCSs for (single) B\"uchi objectives with probability
one~\cite{ABdAMS-fossacs08} and we gave a simplified and generalized proof in~\cite{BS-fmsd2012}. On PLCSs, these positive results cannot be
extended much ---in particular to parity objectives--- since B\"uchi
games with positive probability are undecidable, already in the case
of a single player~\cite{BBS-acmtocl06}.  Attempts to extend the
decidability beyond (generalized) B\"uchi must thus abandon some
generality in other dimensions, e.g., by restricting to finite-memory
strategies, as in the one-player case~\cite{BBS-acmtocl06}.

\paragraph{Outline of the paper.}
Section~\ref{sec-stoch-games} introduces the necessary concepts and
notations on turn-based stochastic games. Section~\ref{sec-charac}
provides the characterization of winning configurations in the general
case of arenas with a finite attractor. Section~\ref{sec-lcs} focuses
on stochastic games on lossy channel systems and explains how
decidability is obtained.



\section{Stochastic games with a finite attractor}
\label{sec-stoch-games}

We consider general 2-player stochastic turn-based games on
finite-choice countable
arenas. In such games, the two players choose moves in turns and the
outcome of their choice is probabilistic.

\begin{definition}
  A \emph{turn-based stochastic arena} is a tuple $\G =
  (\Conf,\Moves,\bfP)$ such that $\Conf$ is a countable set of
  configurations partitioned into $\ConfA \sqcup \ConfB$, $\Moves$ is
  a set of moves, and $\bfP: \Conf \times \Moves \to \Dist(\Conf)$ is
  a partial function whose values are probabilistic distribution of
  configurations.
\\
We say that move $m$ is \emph{enabled} in configuration $c$ when
$\bfP(c,m)$ is defined.
\\
  $\G$ is \emph{eternal} (also \emph{deadlock-free}) if for all $c$ there
is some enabled $m$.
\end{definition}
The set of possible configurations $\Conf$ of the game is partitioned
into configurations ``owned'' by each of the players: in some
$c\in\ConfA$, player $A$, or ``Alice'', chooses the next move, while
if $c\in\ConfB$, it is player $B$, ``Bob'', who chooses.  It is useful
to consider informally that, beyond Alice and Bob, there is a third
party called ``the environment'' who is responsible for the
probabilistic behaviors.  This is why the game is stochastic: after
each move $m$ of one of the players, the environment chooses the next
configuration probabilistically according to $\bfP(c,m)$.
For a configuration $c$, when move $m \in \Moves$ is selected, we
write $\Post[m](c)$ for the set of possible configurations from $c$
after $m$: $\Post[m](c) \egdef \{c'\in \Conf \mid \bfP(c,m)(c') >0\}$,
and, symmetrically, $\Pre[m](c)  \egdef \{c'\in \Conf \mid
\bfP(c',m)(c) >0\}$ denotes the set of possible
predecessors by $m$.

\paragraph{Runs and strategies.}
For simplification purposes, 
we assume in the rest of this paper that all arenas are  eternal,
aka deadlock-free.  A \emph{run} of $\G$ is a (non-empty) sequence
$\rho \in \Conf^* \cup \Conf^\omega$, finite or infinite, of
configurations.  A \emph{strategy} for player $A$ resolves all
non-deterministic choices in $\ConfA$ by mapping every run ending in
an $A$-configuration (i.e., a configuration $c\in\ConfA$) to a move
enabled in $c$. Formally, a strategy $\sigma$ for Alice (an $A$-strategy) is a mapping
$\sigma : \Conf^* \ConfA \to \Moves$ such that, for every history run
$\rho = c_0 c_1 \cdots c_n$ with $c_n \in \ConfA$, $\sigma(\rho)$
is enabled in $c_n$.  Symmetrically, a strategy for Bob (a $B$-strategy)
is a mapping $\tau : \Conf^* \ConfB \to \Moves$ which assigns an
enabled move with each history run ending in $\ConfB$. The pair of
strategies $(\sigma,\tau)$ is called a \emph{strategy profile}. Note
that in this paper we restrict to pure, also called deterministic,
strategies. Allowing for mixed, aka randomized, strategies would not
change the winning configurations~\cite{CdAH-qest04}.

Not all runs agree with a given strategy profile. We say that a finite
or infinite run $\rho = c_0 c_1\cdots c_n \cdots$ is \emph{compatible}
with $(\sigma,\tau)$ if for every prefix $\rho_i = c_0\cdots c_i$ of
$\rho$, $c_i \in \ConfA$ implies $\bfP(c_i,\sigma(\rho_i))(c_{i+1})
>0$, and $c_i \in \ConfB$ implies $\bfP(c_i,\tau(\rho_i))(c_{i+1})
>0$.

\paragraph{Probabilistic semantics.}
The behavior of $\G$ under strategy profile $(\sigma,\tau)$ is
described by an infinite-state Markov chain $\G_{\sigma,\tau}$ where
the states are all the finite runs compatible with $(\sigma,\tau)$, and where
there is a transition from $\rho_i$ to $\rho_{i+1} = \rho_i \cdot
c_{i+1}$ with probability $\bfP(c_i,\sigma(\rho_i))(c_{i+1})$ if
$c_i \in \ConfA$, and $\bfP(c_i,\tau(\rho_i))(c_{i+1})$ if $c_i \in
\ConfB$.  Standardly ---see, e.g., \cite{puterman94} for details---
with the Markov chain $\G_{\sigma,\tau}$ and a starting configuration
$c_0$, is associated a probability measure 
on the set of runs of $\G$ starting with $c_0$ and where behaviors are
ruled by $(\sigma,\tau)$.

It is well-known that given $\varphi$ an LTL formula where
atomic propositions are arbitrary sets of configurations, the set of
runs that satisfy $\varphi$ is
measurable.
Below we write $\PR_{\sigma,\tau}(c_0 \sat \varphi)$ for the
measure of runs of $\G_{\sigma,\tau}$ that start with $c_0$ and
satisfy $\varphi$,  and use the standard ``$\Box$'',
``$\Diamond$'' and ``$\Next$'' symbols for linear-time modalities
``always'', ``eventually'' and ``next''.

\paragraph{Game objectives.}
Given a stochastic arena $\G$, the objective of the game describes the
goal Alice aims at achieving. In this paper we consider generalized B\"uchi
objectives. Let $R_1, \ldots, R_r \subseteq \Conf$ be $r$ sets of
configurations, with an associated generalized B\"uchi property
$\varphi = \bigwedge_{i=1}^r \Box\Diamond R_i$. We consider the game
on $\G$ where Alice's objective is to satisfy $\varphi$ with probability
one.

We say that an $A$-strategy $\sigma$ is \emph{almost-surely
  winning} from $c_0$ for
objective $\varphi$ if for every $B$-strategy  $\tau$,
$\PR_{\sigma,\tau}(c_0 \sat \varphi)=1$. In this case, we say that
configuration $c_0$ is \emph{winning} (for Alice). The set of winning
configurations is denoted $\WinA^{=1} \varphi$, using PATL-like
notation~\cite{chen2007,baier2012b}.

\paragraph{Finite attractor.}
In this paper, we focus on a subclass of stochastic arenas, namely those
with a finite attractor, following a terminology introduced
in~\cite{ABRS-icomp}. We say that a subset $F \subseteq \Conf$ is a
\emph{finite attractor} for the arena $\G$ if (1) $F$ is finite, and (2)
for every initial configuration $c_0$ and for every strategy profile
$(\sigma,\tau)$, $\PR_{\sigma,\tau}(c_0 \models \Box \Diamond F) =1$. In
words, $F$ is almost surely visited infinitely often under all strategy
profiles. Note that an attractor is not what is called a recurrent set in
Markov chains, since ---depending on $c_0$ and $(\sigma,\tau)$--- it does
not necessarily hold that all configurations in $F$ are visited infinitely
often almost surely. An attractor is also not an absorbing set since the players may
leave $F$ after visiting it ---but they will almost surely return to it. Note
also that in game theory one sometimes uses the term ``attractor'' to
denote a set from where one player can ensure to reach a given goal,
something that we call a winning set.
The existence of a finite attractor is a powerful tool for reasoning
about infinite runs in countable Markov chains, see examples
in~\cite{ABRS-icomp,BBS-ipl,rabinovich2006,BBS-acmtocl06,ABdAMS-fossacs08}.

\paragraph{Finite-choice hypothesis.}
Beyond the finite attractor property, we also require that the
adversary, Bob, only has finitely many choice: more precisely, we
assume that in every configuration of $\ConfB$, the set of enabled
moves is finite. Note that we do not assume a \emph{uniform} bound on the
number of moves enabled in Bob's configurations, and also that the
finite-choice hypothesis only applies to Bob, the adversarial player.
These are called \emph{$\diamond$-finitely-branching games} in
\cite{brazdil2013}, and are not a strong restriction in applications,
unlike the finite-attractor
assumption that is usually not satisfied in practice.  We want to stress
that we allow infinite arenas that are \emph{infinitely branching}
both for Alice ---she may have countably many enabled moves in a given
$c$--- and for the environment ---$\Post[m](c)$ may be infinite for
given $c$ and $m$---, and thus are not
\emph{coarse}, i.e., non-zero probabilities are not bounded from
below.



\section{Solving generalized B\"uchi games}
\label{sec-charac}

In this section we provide a simple fixpoint characterization of the set of
winning configurations (and of the associated winning strategies) for games
with a generalized B\"uchi objective that should be satisfied
almost-surely.
For this characterization and its proof of correctness, we use terms with
fixpoints combining functions and constants over the complete lattice
$2^\Conf$ of all sets of configurations.

\subsection{A $\mu$-calculus for fixpoint terms}

We assume familiarity with $\mu$-calculus notation and only recall the
basic concepts and notations we use below. The reader is referred
to~\cite{arnold2001,bradfield2007} for more details.

The set of subsets of configurations ordered by inclusion,
$(2^\Conf,\subseteq)$, is a complete Boolean lattice. We consider
\emph{monotonic operators}, i.e., $n$-ary mappings
$f:(2^\Conf)^n\to(2^\Conf)$ such that $f(U_1,\ldots,U_n)\subseteq
f(V_1,\ldots,V_n)$ when $U_i\subseteq V_i$ for all $i=1,\ldots,n$. (A
\emph{constant} $U\subseteq \Conf$ is a $0$-ary monotonic operator.)
Formally, the language $L_\mu=\{\phi,\psi,\ldots\}$ of \emph{terms
  with fixpoints} is given by the following abstract grammar
\[
 \phi
::=
	       f(\phi_1,\ldots,\phi_n)		\ \big| \
	       X				\ \big| \
	       \mu X.\phi			\ \big| \
	       \nu X. \phi
\]
where $f$ is any $n$-ary monotonic operator and $X$ is any variable.
Terms of the form $\mu X.\phi$ and $\nu X.\phi$ are least and greatest
fixpoint expressions.

The complementation operator $\neg$, defined with $\neg
U=\Conf\setminus U$, may be used as a convenience when writing down
$L_\mu$ terms as
long as any bound variable is under the scope of an even number of
negations. Such terms can be rewritten in positive forms by using the
\emph{dual} $\widetilde{f}$ of any $f$, defined with
$\widetilde{f}(U_1,\ldots,U_n)\egdef \neg f(\neg U_1,\ldots,\neg
U_n)$. Note that $\widetilde{f}$ is monotonic since $f$ is.

The semantics of $L_\mu$ terms is as expected (see~\cite{BS-fmsd2012,bradfield2007}). Since
we only use monotonic operators in our fixpoint terms, all the terms have a
well-defined interpretation as a subset of $\Conf$ for closed terms,
and more generally as a monotonic $n$-ary mapping over $2^\Conf$ for terms
with $n$-free variables. We slightly abuse notation, letting
e.g.\ $\phi(X_1,\ldots,X_n)$ denote both a term in $L_\mu$ and its
denotation as an $n$-ary monotonic operator. Similarly,
$\phi(\psi_1,\ldots,\psi_n)$ is the term obtained by substituting
$\psi_1,\ldots,\psi_n\in L_\mu$ for the (free occurrences of) the $X_i$'s
in $\phi$. Finally, when $U_1,\ldots,U_n\subseteq\Conf$ are constants,
$\phi(U_1,\ldots,U_n)$ also denotes the application of the operator defined by
$\phi$ over the $U_i$'s.

When reasoning on fixpoint terms, one often uses \emph{unfoldings},
i.e., the following equalities stating that a least or
greatest fixpoint is indeed a fixpoint:
\begin{xalignat*}{2}
 \mu X.\phi(X,\ldots)
&=
\phi(\mu X.\phi(X,\ldots), \ldots)
\:,
&
 \nu X.\phi(X,\ldots)
&=
\phi(\nu X.\phi(X,\ldots), \ldots)
\:.
\end{xalignat*}
Recall that the least (or greatest) fixpoint is the least \emph{pre-fixpoint}
(greatest \emph{post-fixpoint}):
\begin{xalignat*}{2}
\phi(U)\subseteq U &\mbox{ implies } \mu X.\phi(X) \subseteq U
\:,
&
\phi(U)\supseteq U &\mbox{ implies } \nu X.\phi(X) \supseteq U
\:.
\end{xalignat*}

It is well-known (Kleene's fixpoint theorem) that when monotonic
operators are $\bigcup$- and $\bigcap$-continuous, ---i.e., satisfy
$f(\bigcup_{i} U_i)=\bigcup_{i} f(U_i)$ and $f(\bigcap_{i}
U_i)=\bigcap_{i} f(U_i)$---, their least and greatest fixpoints are
obtained as the limits of $\omega$-length sequences of approximants.
Since we do not assume $\bigcap/\bigcup$-continuity in our setting ---e.g.,
$\Pre$ is not $\bigcap$-continuous when finite-branching is not
required---, fixpoints are obtained as the stationary limits of
transfinite ordinal-indexed sequences of approximants,
see~\cite{bradfield2007}. For a set $U=\mu X.\phi(X)$ defined as a
least fixpoint, the approximants $(U_\alpha)_{\alpha\in\Ord}$ are
defined inductively with $U_0\egdef\emptyset$,
$U_{\beta+1}\egdef\phi(U_\beta)$ for a successor ordinal, and
$U_\lambda\egdef \bigcup_{\beta<\lambda}U_\beta$ for a limit ordinal
$\lambda$. For a greatest fixpoint $V=\nu X.\phi(X)$, they are given
by $V_0\egdef \Conf$, $V_{\beta+1}\egdef\phi(V_\beta)$ and
$V_\lambda\egdef \bigcap_{\beta<\lambda}V_\beta$.

\subsection{A characterization of winning sets}

We first introduce auxiliary operators that let us reason about
strategies and characterize the winning sets. Let
$\Ena(c)\subseteq\Moves$ denote the set of moves enabled in
configuration $c$ and for $X,Y \subseteq \Conf$ let
\begin{align}
\notag
\Pexists(X,Y) &\egdef \{c \in \Conf \mid \exists m \in \Ena(c), \Post[m](c)
\subseteq X \textrm{ and } \Post[m](c) \cap Y \neq \emptyset\} \:,\\
\notag
\Pforall(X,Y) &\egdef \{c \in \Conf \mid \forall m \in \Ena(c), \Post[m](c)
\subseteq X \textrm{ and } \Post[m](c) \cap Y \neq \emptyset\}\:.
\end{align}
One can see that $\Pexists$ and $\Pforall$ are monotonic in both
arguments by reformulating their definitions in terms of the more
familiar $\Pre$ operator (recall that $c\in\wPre[m](\emptyset)$ iff
$m$ is not enabled in $c$):
\begin{align*}
\Pexists(X,Y) &= \bigcup_{\!\!\!\!\!\!\!\!m \in \Moves\!\!\!\!\!\!\!\!} \bigl[\wPre[m](X) \cap
\Pre[m](Y)\bigr]\:,
\\
\Pforall(X,Y) &= \bigcap_{\!\!\!\!\!\!\!\!m \in
  \Moves\!\!\!\!\!\!\!\!}
\bigl(\wPre[m](\emptyset)\cup\bigl[\wPre[m](X)\cap \Pre[m](Y)\bigr]\bigr)\:.
\end{align*}
We further define
\(
\PxA(X,Y) \egdef \bigl(\ConfA \cap \Pexists(X,Y)\bigr) \cup \bigl(\ConfB \cap
\Pforall(X,Y)\bigr).
\)
In other words, $\PxA(X,Y)$ is exactly the set from where Alice can guarantee in one step
to have $X$ surely and $Y$ with positive probability. This can be
summarized as:
\begin{fact}
\label{fact:pxa}
Let $X,Y \subseteq \Conf$.
\\
1.\ If $c \in \PxA(X,Y)$, then, $A$ has a
memoryless strategy $\sigma$ such that, for every strategy $\tau$ for
$B$: $\PR_{\sigma,\tau}(c \models \Next X) =1$ and
$\PR_{\sigma,\tau}(c \models \Next Y) >0$.
\\
2.\ If $c \notin \PxA(X,Y)$, then, $B$ has
a memoryless strategy $\tau$ such that, for every strategy $\sigma$
for $A$: $\PR_{\sigma,\tau}(c \models \Next X) <1$ or
$\PR_{\sigma,\tau}(c \models \Next Y) =0$.
\end{fact}

\noindent
Building on $\PxA$, we may define the following unary operators: for $i=1,\ldots,r$, $H_i$ is
given by
\begin{align}
\label{eq-def-Hi}
 H_i(X)&\egdef\mu Z.X\cap \PxA\bigl(X,R_i\cup Z\bigr)
\:.
\end{align}
The intuition is that, from $H_i(X)$, Alice has a strategy ensuring a
positive probability of reaching $R_i$ later ---which would be
characterized by ``$\mu Z.\PxA\bigl(\Conf,R_i \cup Z\bigr)$''---
\emph{all the while staying surely in $X$}, hence the amendments. See
Lemma~\ref{lm:charac-Hi} for a precise statement.  Unfolding its definition, we see
that $H_i(X)\subseteq X$, i.e., $H_i$ is \emph{contractive}.

Letting $\HA(X) \egdef \bigcap_{i=1}^r H_i(X)$, we finally define the
following fixpoint terms:
\begin{align}
\label{eq-def-W}
W &\egdef \nu X. \HA(X) 
= \nu X.\bigcap_{i=1}^r\Bigl[\mu Z.X\cap\PxA\bigl(X,R_i\cup Z\bigr)\Bigr] \:,
\\
\label{eq-def-W'}
W' &\egdef \nu X.\PxA\Bigl(\HA(X),\Conf\Bigr) = \nu
X.\PxA\Bigl(\bigcap_{i=1}^r \Bigl[\mu Z.X\cap\PxA\bigl(X,R_i\cup Z\bigr)\Bigr],\Conf\Bigr) \:,
\\
\label{eq-def-W1}
W_1 &\egdef \nu X.\mu Z.\PxA\bigl(X,R_1\cup Z\bigr).
\end{align}

\begin{theorem}[Fixpoint characterization of winning sets]
\label{th-correct-W}
We fix a stochastic arena with a finite attractor, and assume it is
finite-choice for Bob. Then, for generalized B\"uchi objectives the
winning set $\WinA^{=1} \bigwedge_{i=1}^r \Box\Diamond R_i$ coincides
with $W$.  Moreover $W=W'$ and from $W$ Alice has an almost-surely
winning strategy $\sigma_W$ that is a \emph{finite-memory} strategy.

\noindent
In the case $r=1$ of simple B\"uchi objectives
the winning set $\WinA^{=1} \Box\Diamond R_1$
coincides with $W_1$ and the winning strategy $\sigma_W$ is even a \emph{memoryless} strategy.
\end{theorem}
Before proving Theorem~\ref{th-correct-W},
let us explain how, in the case where $r=1$, one derives the
correctness of $W_1$ from the correctness of $W$. Setting $r=1$ in
Eq.~\eqref{eq-def-W}
yields $W = \nu X.\mu Z.X\cap
\PxA\bigl( X,R_1\cup Z\bigr)$. In this situation, we can use
Eq.~\eqref{eq-mu-nu-contractive}, a purely algebraic and
lattice-theoretical equality that holds for any monotonic binary $f$
(see Appendix for a proof):
\begin{gather}
\label{eq-mu-nu-contractive}
\tag{$\dagger$} \nu X.\mu Z.X\cap f(X,Z) = \nu X.\mu Z. f(X,Z) \:.
\end{gather}
Applying Eq.~\eqref{eq-mu-nu-contractive} on $W= \nu X.\mu Z.X\cap
\PxA\bigl( X,R_1\cup Z\bigr)$ yields
$W = \nu X.\mu Z. \PxA\bigl(X,R_1\cup Z\bigr)=W_1$.

Theorem~\ref{th-correct-W} provides two different characterizations of
the winning set $\WinA^{=1}\bigl(\bigwedge_{i=1}^r\Box\Diamond
R_i\bigr)$. Let us now prove its validity, in the general context of
finite-choice stochastic arenas with a finite attractor\footnote{We
  show later that the characterizations of the winning sets is
  \emph{not correct} if one does not assume the finite attractor
  property.}. The proof is divided in two parts: correctness of $W'$ in
Proposition~\ref{prop-W-correct}, completeness of $W$ in
Proposition~\ref{prop-compl-U}, and some purely lattice-theoretical reasoning
closing the loop in Lemma~\ref{lem-U=W}.

\subsection{Correctness for $W'$}

We prove that $W'$ only contains winning configurations for Alice by
exhibiting a strategy with which she ensures almost surely
$\bigwedge_i\Box\Diamond R_i$ when starting from some $c\in W'$. We
first define $r$ strategies $(\sigma_i)_{1\leq i\leq r}$, one for each
goal set $R_1,\ldots,R_r$, and prove their relevant properties. It
will then be easy to combine the $\sigma_i$'s in order to produce the
required strategy.

For $i=1,\ldots,r$, unfolding Eq.~\eqref{eq-def-Hi} yields $H_i(W') =
W' \cap \PxA(W',R_i \cup
H_i(W'))$. We let $\sigma_i$ be the memoryless $A$-strategy defined as
follows: for $c\in\ConfA\cap H_i(W')$, Alice picks an enabled move $m$
such that $\Post(c)[m] \subseteq W'$ and $\Post[m](c) \cap (R_i \cup
H_i(W')) \neq \emptyset$, which is possible by definition of $\PxA$,
while for $c\in\ConfA\cap W'\cap\neg H_i(W')$, Alice picks an enabled
move $m$ with $\Post(c)[m] \subseteq \HA(W')$, which is possible since
$W'=\PxA(\HA(W'),\Conf)$ by Eq.~\eqref{eq-def-W'}.

$\HA$ is contractive since the $H_i$'s are, hence $\HA(W')\subseteq W'$ and we deduce
that ``$\sigma_i$ stays in $W'$'':
\begin{gather}
\label{eq-sigmai-BoxW}
\forall c\in W':
\forall \tau:
\PR_{\sigma_i,\tau}(c\sat\Box W')=1
\:.
\end{gather}

\begin{lemma}
\label{lem-exists-gamma-c}
For all $c\in W'$ there exists some $\gamma_c>0$ such that
$\PR_{\sigma_i,\tau}(c \sat \Diamond R_i) \geq \gamma_c$ for all
$B$-strategies $\tau$.
\end{lemma}
\begin{proof}
  Here we use the finite-choice assumption.  First consider the case
  where $c\in H_i(W')$. Writing
  $(Z_\alpha)_{\alpha\in\Ord}$ for the approximants of $H_i(W')=\mu Z.W'\cap \PxA\bigl(W',R_i\cup Z\bigr)$,
  we
  prove, by induction on $\alpha$, that $\gamma_c>0$ exists when $c\in
  Z_\alpha$. The base case $\alpha=0$ holds vacuously since
  $Z_0=\emptyset$.  For $\alpha=\lambda$ (a limit ordinal),
  $Z_\lambda=\bigcup_{\beta<\lambda}Z_\beta$ so each $c\in Z_\lambda$
  is in some $Z_\beta$ and the induction hypothesis applies.

  Now to the successor case $\alpha=\beta+1$. Here
  $Z_\alpha=W'\cap\PxA\bigl(W',R_i\cup Z_{\beta}\bigr)$ and, given
  $\sigma_i$ and for any $\tau$, from $c\in Z_\alpha$ Alice or Bob
  will pick a move $m$ with $\Post[m](c) \cap (R_i\cup Z_{\beta}) \neq
  \emptyset$. The probability that after probabilistic environment's
  move the play will be in $R_i$ exactly at the next step is precisely
  $\gamma=\sum_{d\in R_i} \bfP(c,m)(d)$ and $\gamma>0$ if $\Post[m](c)
  \cap R_i\neq \emptyset$ (and only then). If $\gamma=0$ then
  $\Post[m](c) \cap R_i = \emptyset $ so that $\Post[m](c) \cap
  Z_\beta \neq \emptyset $. Then there is a positive probability
  $\gamma'=\sum_{d\in Z_\beta}\bfP(c,m)(d)$ that after probabilistic
  decision the play will be in $Z_{\beta}$ at the next step, hence (by
  induction hypothesis) a positive probability $\gamma''$ that it will
  be in $R_i$ later, with $\gamma''\geq \sum_{d\in
    Z_\beta}\gamma_{d}\cdot \bfP(c,m)(d)$. Note that for $d\in
  Z_\beta$, $\gamma_d$ does not depend on $\tau$ (by ind.\ hyp.) so
  that $\gamma$ and the lower bound for $\gamma''$ only slightly
  depend on $\tau$: they depend on what move $m$ is chosen by Bob if
  $c\in\ConfB$.  Now, since there are only finitely many moves enabled
  in $c$, we can pick a strictly positive value that is a lower bound
  for all the corresponding $\max(\gamma,\gamma'')$, proving the
  existence of $\gamma_c>0$ for $c\in Z_\alpha$.

  There remains the case where $c\in W'\cap\neg H_i(W')$: here
  $\sigma_i$ ensures that the play will be in $H_i(W')$ in the next
  step. If $c\in \ConfB$, we can let
  $\gamma_c\egdef \min_{m \in \Ena(c)}\sum_{d\in H_i(W')}\bfP(c,m)(d)
  \cdot \gamma_{d}$, which ensures $\gamma_c >0$ by the finite-choice
  assumption. In case $c \in \ConfA$, we simply define $\gamma_c\egdef
  \sum_{d\in H_i(W')}\bfP(c,m)(d) \cdot \gamma_{d}$ where $m$ is the
  move given by $\sigma_i$ when in configuration $c$. In both cases,
  we thus have $\gamma_c >0$, which concludes the proof.
\end{proof}
\begin{remark}[On the finite-choice assumption for $B$]
  Clearly enough, Lemma~\ref{lem-exists-gamma-c} does not hold if we
  relax the assumption that in every configuration of $\ConfB$, the
  set of enabled moves is finite. Indeed, consider a simple arena with
  three configurations $c$, $r$ and $s$, all belonging to player $B$,
  where $r$ and $s$ are sinks, and from $c$ there are countably many
  enabled moves $m_1, m_2 \cdots$ whose respective effect is defined
  by $\bfP(c,m_k)(r) = 1/2^k$ and $\bfP(c,m_k)(s) = 1- 1/2^k$. Letting
  $R = \{r\}$ and for the single B\"uchi objective $\Box\Diamond R$,
  we obtain $W' = \Conf$, and in particular $c \in W'$. Yet, there is
  no uniform lower bound $\gamma_c$ with $\PR_{\sigma_i,\tau}(c \sat
  \Diamond R_i) \geq \gamma_c$ for all $B$-strategies $\tau$.
\end{remark}

\begin{lemma}
\label{lem-sigmai-BoxDiamondRi=1}
$\PR_{\sigma_i,\tau}(c \sat \Box W' \land \Box\Diamond
R_i)=1$ for all $c\in W'$ and all $B$-strategies $\tau$.
\end{lemma}
\begin{proof}
  This is where we use the finite-attractor property: there is a
  finite set $F \subseteq \Conf$ such that $\PR_{\sigma,\tau} (c \sat
  \Box \Diamond F)=1$ for any $c\in\Conf$ and strategies $\sigma$ and
  $\tau$. In particular, for $\sigma_i$ and using
  Eq.~\eqref{eq-sigmai-BoxW}, we deduce $\PR_{\sigma_i,\tau} (c \sat
  \Box W'\land\Box \Diamond F)=1$ for any $c\in W'$ and any strategy
  $\tau$ (entailing $F\cap W'\not=\emptyset$). Let now $\gamma\egdef
  \min\{ \gamma_f ~|~ f\in F\cap W'\}$ so that for any $f\in F\cap W'$
  and any $B$-strategy $\tau$, Lemma~\ref{lem-exists-gamma-c} gives
  $\PR_{\sigma_i,\tau}(f \sat \Diamond R_i) \geq \gamma$. Note that
  $\gamma>0$ since $F\cap W'$ is finite. Since from $F \cap W'$ and
  applying $\sigma_i$, the probability to eventually reach $R_i$ is
  lower bounded by $\gamma$, and since $\PR_{\sigma_i,\tau} \bigl(c
  \sat \Box\Diamond (F\cap W')\bigr)=1$, we deduce that
  $\PR_{\sigma_i,\tau} (c \sat \Box\Diamond R_i)=1$ by standard
  reasoning on recurrent sets.
\end{proof}

\begin{remark}[On the finite-attractor assumption]
Lemma~\ref{lem-sigmai-BoxDiamondRi=1} crucially relies on the
finite-attractor
property. Indeed, consider the random walk on the set of
naturals where from any state $n>0$ the probability is $\frac{3}{4}$
to move to $n+1$ and $\frac{1}{4}$ to move to $n-1$ (and where state
$0$ is a sink where one stays forever). It is a well-known result
that, starting from any $n >0$, the probability is strictly less than
$1$ to visit state $0$ ---and, in fact, any finite set of states---
infinitely often. This random walk however can be seen as a stochastic
game (with a single player and a single move in each state) for which,
and taking $R_1 = \{0\}$, $W_1$ consists of the whole  states set
(indeed,  $W_1=\Pre^*(R_1)$ for  single-player single-choice arenas).
This provides a simple example showing that the
finite-attractor
property is required for
Lemma~\ref{lem-sigmai-BoxDiamondRi=1}, and for
Theorem~\ref{th-correct-W}, to hold.
\qed
\end{remark}

\begin{proposition}[Correctness of $W'$]
\label{prop-W-correct}
$W' \subseteq \WinA^{=1}\bigl(\bigwedge_{i=1}^r \Box\Diamond R_i\bigr)$.
\end{proposition}
\begin{proof}
By combining the strategies $\sigma_i$'s, we define a finite-memory strategy
$\sigma_W$ that guarantees $\PR_{\sigma_W,\tau}\bigl(c\sat\bigwedge_i
\Box\Diamond R_i\bigr)=1$ for any $c\in W'$ and against any $B$-strategy
$\tau$.

More precisely, $\sigma_W$ has $r$ modes: $1,2,\ldots,r$. In mode $i$,
$\sigma_W$ behaves like $\sigma_i$ until $R_i$ is reached, which is bound to
eventually happen with probability 1 by
Lemma~\ref{lem-sigmai-BoxDiamondRi=1}. Note that the play remains
constantly in $W'$. Once $R_i$ has been reached, $\sigma_W$ switches to mode
$i+1(\text{mod }r)$, playing at least one move. This is repeated in a
neverending cycle, ensuring $\Box\Diamond R_i$ with probability 1.
\end{proof}

\begin{remark}[On randomized memoryless strategies]
It is known that, if one considers mixed, aka randomized, strategies,
generalized B\"uchi objectives on $\G$ admit memoryless, aka
deterministic, winning strategies~\cite{brazdil2013}. Note that our
$\sigma_W$ is finite-memory (and not randomized). It is a natural
question whether a simple randomized memoryless strategy like ``at
each step, choose randomly and uniformly between following $\sigma_1,
\ldots, \sigma_r$'' is almost-surely winning for
$\bigwedge_{i=1}^r\Box\Diamond R_i$.
\end{remark}

\subsection{Completeness of $W$}

In order to prove that $W$ 
contains the winning set for Alice, we show that $\WinA^{=1}
\bigwedge_i\Box\Diamond R_i$ is a post-fixpoint of $\HA$, thus
necessarily included in its greatest fixpoint $W$. We start with the
following lemma:
\begin{lemma}
\label{lm:charac-Hi}
\(H_i(X) \supseteq \bigl\{c ~\big|~ \exists \sigma\ \forall \tau,\
\PR_{\sigma,\tau}(c \sat \Box X)=1 \textrm{ and }
 \PR_{\sigma,\tau}(c \sat \Next\Diamond R_i) >0\bigr\}\:.\)
\end{lemma}

\begin{proof}
We actually prove a stronger claim: we show that there exists a memoryless
$B$-strategy $\tau$  such that,  for every $c \notin H_i(X)$ and every $A$-strategy $\sigma$, either $\PR_{\sigma,\tau}(c\models \Diamond \neg X) >0$, or
$\PR_{\sigma,\tau}(c\models \Next\Box \neg R_i)=1$.

Let $c \notin H_i(X)$. By definition $\neg H_i(X) = \neg X \cup \neg
\PxA(X,R_i \cup H_i(X))$. If $c \notin X$, then trivially
$\PR_{\sigma,\tau}(c \models \Diamond \neg X) >0$ for any
$(\sigma,\tau)$ so we do not care how $\tau$ is defined here. Consider
now $c \notin \PxA(X,R_i \cup H_i(X))$. By Fact~\ref{fact:pxa}, Bob
has a (memoryless) strategy $\tau_c$ such that against any $A$-strategy
$\sigma$, $\PR_{\sigma,\tau_c}(c \models \Next X)<1$ or
$\PR_{\sigma,\tau_c}\bigl(c \models \Next (R_i \cup H_i(X))\bigr)=0$,
which can be reformulated as $\PR_{\sigma,\tau_c}(c \models \Next \neg
X)>0$ or $\PR_{\sigma,\tau_c}\bigl(c \models \Next (\neg R_i \cap \neg
H_i(X))\bigr)=1$. For $c\in\ConfB$, we define $\tau(c)$ as the move
given by $\tau_c(c)$. The resulting strategy $\tau$ guarantees,
starting from $\neg H_i(X)$, that the game will either always stay in $\neg
R_i\cap\neg H_i(X)$ (after the 1st step) or has a positive probability
of visiting $\neg X$ eventually.
\end{proof}

\begin{proposition}[Completeness of $W$]
\label{prop-compl-U}
$\WinA^{=1} \bigl(\bigwedge_{i=1}^r\Box\Diamond R_i\bigr)\subseteq W$.
\end{proposition}
\begin{proof}
  Let $c \in \WinA^{=1} \bigwedge_i \Box\Diamond R_i$, and $\sigma$ be
  a strategy ensuring $\bigwedge_i\Box\Diamond R_i$ with probability 1
  from $c$. Consider $E = \bigl\{d \in \Conf \mid \exists \tau :
  \PR_{\sigma,\tau}(c \sat \Diamond d)>0\bigr\}$, i.e., the set of
  configurations that can be visited under strategy
  $\sigma$. Obviously $c\in E$. Furthermore, for any $d\in E$ and any
  $B$-strategy $\tau$, $\PR_{\sigma',\tau}(d\sat\Box E)=1$ holds,
  where $\sigma'$ is a ``suffix strategy'' of $\sigma$ after $d$ is
  visited, that is, $\sigma'$ behaves from $d$ like $\sigma$ would
  after some prefix ending in $d$. Since furthermore
  $\PR_{\sigma',\tau}\bigl(d\sat\bigwedge_i\Box\Diamond R_i\bigr)=1$
  by assumption, we deduce in particular
  $\PR_{\sigma',\tau}\bigl(d\sat\Next\Diamond R_i\bigr)=1$ for any
  $i=1,\ldots,r$. Hence $E\subseteq H_i(E)$ for any $i$ by
  Lemma~\ref{lm:charac-Hi}, and thus $E\subseteq \HA(E)$. Finally $E$
  is a post-fixpoint of $\HA$, and is thus included in its greatest
  fixpoint.  We conclude that $c\in\nu X.\HA(X)= W$.
\end{proof}
The loop is closed, and Theorem~\ref{th-correct-W} proven, with the
following lattice-theoretical reasoning:
\begin{lemma}
\label{lem-U=W}
$W \subseteq W'$.
\end{lemma}
\begin{proof}
  Recall that $W = \nu X. \HA(X)$, so that $W=\HA(W)$. Similarly,
  using Eq.~\eqref{eq-def-Hi}, we deduce $\HA(W)=\bigcap_i H_i(W)=
  \bigcap_i \bigl(W\cap\PxA(W,R_i \cup H_i(W))\bigr)$, hence
  $\HA(W)\subseteq\PxA(W,\Conf)$ by monotonicity of $\PxA$ in its
  second argument. Combining these two points gives
  $W=\HA(W)\subseteq\PxA(\HA(W),\Conf)$, hence $W$ is a post-fixpoint
  of $X \mapsto \PxA(\HA(X),\Conf)$ and is included in its greatest
  fixpoint $W'$.
\end{proof}



\section{Stochastic games on lossy channel systems}
\label{sec-lcs}

Theorem~\ref{th-correct-W} entails the decidability of
generalized B\"uchi games on channel systems with probabilistic message
losses, or PLCSs. This is obtained by applying a generic and powerful
``finite-time convergence theorem'' for fixpoints defined on WQO's.

\subsection{Channel systems with guards}

A \emph{channel system} 
is a tuple $S= (Q,\Ch,\Mess,\Delta)$ consisting of a finite set
$Q=\{q,q',\ldots\}$ of \emph{locations}, a finite set
$\Ch=\{\ttc_1,\ldots,\ttc_d\}$ of \emph{channels}, a finite
\emph{message alphabet} $\Mess=\{a,b,\ldots\}$ and a finite set
$\Delta=\{\delta,\ldots\}$ of \emph{transition rules}.  Each
transition rule has the form $(q,g,\op,q')$, written $q \step{g,\op}
q'$, where $g$ is a \emph{guard} (see below), and $\op$ is an
\emph{operation} of one of the following three forms: $\ttc!a$
(sending message $a \in \Mess$ along channel $\ttc \in \Ch$), $\ttc?a$
(receiving message $a$ from channel $\ttc$), or $\surd$ (an internal
action with no I/O-operation).

Let $S$ be a channel system as above. A \emph{configuration} of $S$ is
a pair $c=(q,w)$ where $q$ is a location of $S$ and $w : \Ch \to
\Mess^*$ is a mapping, that describes the current channel contents:
we let $\ConfS\egdef Q \times {\Mess^*}^\Ch$.

A \emph{guard} is a predicate on channel contents used to constrain
the firability of rules. In this paper, a guard is a tuple
$g=(L_1,\ldots,L_d)\in\Reg(\Mess)^{\sizeof{\Ch}}$ of regular
languages, one for each channel.
For a configuration
$c=(q,w_1,\ldots,w_d)$, we write $c\sat g$, and say that \emph{$c$
  respects $g$}, when $w_i\in L_i$ for all $i=1,\ldots,d$.

Rules give rise to transitions in the operational
semantics. Let $\delta=(q_1,g,\op,q_2)$ be a rule in $\Delta$ and let
$c=(q,w)$, $c'=(q',w')$ be two configurations of $S$. We write $c
\step{\delta} c'$, and say that $\delta$ is enabled in $c$, if $q=q_1$,
$q'=q_2$, $c\sat g$, and $w'$ is the valuation obtained from $w$ by
applying $\op$. Formally $w' = w$ if $\op = \surd$, and otherwise if
$\op=\ttc_i!a$ (resp. if $\op=\ttc_i?a$) then $w'_i=w_i.a$
(resp. $a.w'_i=w_i$) and $w'_j=w_j$ for all $j\not=i$.

For simplicity, we assume in the rest of the paper that $S$ denotes a
fixed channel system $S=(Q,\Ch,\Mess,\Delta)$  that  has
no deadlock configurations, i.e., every $c\in\ConfS$ has an enabled
rule: this is no loss of generality since it is easy ---\emph{when guards are
allowed}--- to add rules going to a new sink location exactly in
configurations where none of the original rules is enabled.

\begin{remark}[About guards in channel systems]
Allowing guards in transition rules is useful (e.g., for expressing
priorities) but departs from the standard models of channel
systems~\cite{muscholl2010}.  Indeed, testing the
whole contents of a fifo channel is not a realistic feature  when modeling
distributed asynchronous systems. However, (unreliable) channel systems are now seen
more broadly as a fundamental computational model closely related to
Post's \emph{tag systems} and with algorithmic
applications beyond distributed protocols:
see, e.g.,~\cite{konev2005,abdulla-icalp05,BFL-lics2012,bouyer2008}.
In such settings,  simple guards have been considered and proved useful:
see, e.g.,~\cite{bouyer2008,BMOSW-stacs08,JKS-tcs2012}.

Using additional control states and messages, it is sometimes possible to simulate guards in (lossy) channel
systems. We note that the known
simulations  preserve nondeterministic reachability but usually not
game-theoretical properties in stochastic environments.
\qed
\end{remark}

\subsection{Probabilistic message losses}
\label{sec:plosses}

PLCSs are channel systems where messages can be lost (following some
probabilistic model) while they are in the
channels~\cite{purush97,baier99,ABRS-icomp,abdulla2005,rabinovich2006}.
In this paper, we consider two kinds of unreliability caused by a
stochastic environment: message losses on one hand, and combinations of
message losses and duplications on the other hand.

Message losses are traditionally modeled via the subword relation: given
two words $u,v\in\Mess^*$, we write $u\subword v$ when $u$ is a
\emph{subword}, i.e., a scattered subsequence, of $v$. For two
configurations $c=(q,w)$ and $c'=(q',w')$, we let $c\subword c'$
$\equivdef$ $\bigl(q=q'$ and $w_i\subword w'_i$ for all $i=1,\ldots,d\bigr)$. In
other words, $c\subword c'$ when $c$ is the result of removing some
messages (possible none) at arbitrary places in the channel contents for
$c'$.

Message duplications are modeled by a rational transduction
$\Tdup\subseteq\Mess^*\times\Mess^*$ over sequences of messages, where
every single message $a\in\Mess$ is replaced by either $a$ or $a a$. We
write $u\preceq_\dup v$ when $(u,v)\in \Tdup$ (\emph{e.g.} $ab
\preceq_\dup aab$) 
and we extend to configurations with $(q,w)\preceq_\dup
(q',w')\equivdef$ $\bigl(q=q'$ and $w_i\preceq_\dup w'_i$ for all
$i=1,\ldots,d\bigr)$.

For PLCSs with only message losses, we write $c\leadsto c'$ when
$c\supword c'$ ($\equivdef c'\subword c$). For PLCSs with losses and
duplications, $c\leadsto c'$ means that $c\preceq_\dup c''\supword c'$ for
some $c''$.

In PLCSs, message perturbations are probabilistic events. Formally,
we associate a distribution $D_\env(c)\in \Dist(\ConfS)$ with every
configuration $c\in\ConfS$ and we say that ``$D_\env(c,c')$ is the
probability that $c$ becomes $c'$ by message losses and duplications
(in one step)''. Given $D_\env$ and a partition $\ConfS = \ConfA
\sqcup \ConfB$, the channel system $S$ with probabilistic losses
defines a stochastic arena $\G_S = (\ConfS,\Delta,\bfP)$ where the
moves available to the players are exactly the rules of $S$ ---thus $\G_S$
is finite-choice---, and the probabilistic transition function $\bfP$ is
formalized by: for every $c \in \ConfS$ and $\delta$ enabled in $c$,
$\bfP(c,\delta) \egdef D_\env(c')$ where $c \step{\delta} c'$.

The qualitative properties that  we are interested in
 do not depend on the exact choices made for $D_\env$. In this paper,   we only
require that  $D_\env$ is \emph{well-behaved}, i.e., satisfies the
following two properties:
\begin{description}
\item[Compatibility with nondeterministic semantics:]
$D_\env(c)(c')>0$ iff $c\leadsto c'$.
\item[Finite attractor:]
Some finite set $F\subseteq\ConfS$ is visited infinitely often
with probability one.
\end{description}
A now standard choice for $D_\env$ in PLCSs models message losses (and
duplications) as independent events. One assumes that at every step, each
individual message can be lost with a fixed probability $\lambda \in
(0,1)$, duplicated with a fixed probability $\lambda'\in[0,1)$ (and remains
  unperturbed with probability $1-\lambda-\lambda'$
). This is the so-called \emph{local-fault} model
  from~\cite{ABRS-icomp,rabinovich2006,Sch-voss}, and it gives rise to a
  well-behaved $D_\env$ when only message losses are considered, i.e., when
  $\lambda'=0$, or when losses are more probable than duplications, i.e.,
  when $0<\lambda'<\lambda$.
  In particular, the set $F_0 \egdef \{ (q,\epsilon,\ldots,\epsilon)
  ~|~ q\in Q\}$ of configurations with empty channels is a finite
  attractor in $\G_S$.
  The interested reader can find in \cite[sections~5\&6]{ABRS-icomp}
  some detailed computations of $D_\env(c)(c')$ in the local-fault
  model, but s/he must be warned that the qualitative outcomes on
  PLCSs do not depend on these values as long as $D_\env$ is
  well-behaved.

\subsection{Regular model-checking of channel systems}

Regular model-checking~\cite{bouajjani2000b,kesten2001} is a symbolic
verification technique where one
computes infinite but regular sets of configurations using representations
from automata theory or from constraint solving.

\begin{definition}
A \emph{(regular) region} of $S$ is a set
$R\subseteq\ConfS$ of configurations that can be written under the form $R
=\bigcup_{i\in I} \{q_i\}\times L_i^1\times \cdots \times L_i^d$ with
 a \emph{finite} index set $I$, and where, for $i\in I$, $q_i$ is some
location $\in Q$, and each $L_i^j$ for $j=1,\ldots,d$ is a regular language
$\in\Reg(\Mess)$.
\end{definition}
Let $\Rcal\subseteq 2^\ConfS$ denote the set of all regions of $S$. A
monotonic operator $f$ is \emph{regularity-preserving}, if
$f(R_1,\ldots,R_n)\in\Rcal$ when $R_1,\ldots,R_n\in\Rcal$. A
regularity-preserving $f$ is \emph{effective} if a representation for
$f(R_1,\ldots,R_n)$ can be computed uniformly from representations for the
$R_i$'s (and $S$). For example, the set-theoretical $\cap$, $\cup$ are
regularity-preserving and effective. While not a monotonic operator,
complementation is regularity-preserving and effective. Hence the dual
$\widetilde{f}$ of any $f$ is regularity-preserving and effective when $f$
is.

For the verification of (lossy) channel systems in general, and the resolution
of  games in particular, some useful operators are the unary
pre-images $\Pre_S[\delta]$ for $\delta\in\Delta$, and the upward- and
downward-closures $\Cu$ and $\Cd$, defined with
\begin{xalignat*}{2}
   \Pre_S[\delta](U) &\egdef \{c\in\ConfS~|~\exists c'\in U: c\step{\delta} c'\}\:,
&  \Cu(U) &\egdef \{c\in\ConfS~|~\exists c'\in U: c'\subword c\}\:,
\\
   \Pre_S(U) &\egdef {\textstyle \bigcup_{\delta\in\Delta}\Pre_S[\delta](U)}\:,
&  \Cd(U) &\egdef \{c\in\ConfS~|~\exists c'\in U: c\subword c'\}\:.
\end{xalignat*}
Observe that $\Pre_S[\delta]$ and $\Pre_S$ are pre-images for steps of channel systems
without/before message perturbations, while $\Cu$ and $\Cd$ are pre- and
post-images for the message-losing relation.
$\Cu$ and $\Cd$ are closure operators. Their duals are \emph{interior
  operators}: $\Kd(U)\egdef\widetilde{\Cu}(U)$ and
$\Ku(U)\egdef\widetilde{\Cd}(U)$ are the largest downward-closed and, resp.,
upward-closed, subsets of $U$.
Finally, we are also interested in pre-images for $\preceq_\dup$: we write
$\invTdup(U)$ for $\{c~|~\exists c'\in U: c\preceq_\dup c'\}$. We remark that
$\invTdup(\ConfS) = \ConfS$, and that $\invTdup(\Cu U)=\Cu(\invTdup
(\Cu U))=\Cu(\invTdup(U))$, i.e., the definition
of $c\leadsto c'$ is not sensitive to the order of perturbations.
\begin{fact}
$\Pre_S[\delta]$, $\Pre_S$, $\Cu$, $\Cd$, $\invTdup$ and their duals are
  regularity-preserving and effective (monotonic) operators.
\end{fact}

When using effective regularity-preserving operators, one can evaluate any
closed $L_\mu$ term that does not include fixpoints. For a closed term $U=\mu
X.\phi(X)$, or $V=\nu X.\phi(X)$, with a single fixpoint, any
approximant $U_k$ and $V_k$ for a finite $k\in\Nat$ can be evaluated but
there is no guarantee that the fixpoint is reached in finite time, or that
the fixpoint is a regular region.
However, for fixpoints over a WQO like $\ConfS$, there exists a generic finite-time
convergence theorem.

\begin{definition}[Guarded $L_\mu$ terms]
1.~A variable $Z$ is \emph{upward-guarded} in an $L_\mu$ term $\phi$ if every
occurrence of $Z$ in $\phi$ is under the scope of 
an upward-closure $\Cu$ or upward-interior $\Ku$ operator.

\noindent
2.~It is \emph{downward-guarded}
in $\phi$ if all its  occurrences in $\phi$ are under the scope of a
downward-closure $\Cd$ or downward-interior $\Kd$ operator.

\noindent
3.~A term $\phi$ is \emph{guarded} if every least fixpoint subterm $\mu
Z.\psi$ of $\phi$ has $Z$ upward-guarded in $\psi$, and every
greatest fixpoint subterm $\nu Z.\psi$ has $Z$ downward-guarded in $\psi$.
\end{definition}
\begin{theorem}[Effective \& regularity-preserving fixpoints]
\label{th-guarded-effective}
Any guarded $L_\mu$ term $\phi(X_1,\ldots,X_n)$
built with regularity-preserving and effective operators
denotes a regularity-preserving and effective $n$-ary operator.
Furthermore the denotation of a closed term can be evaluated by computing
its approximants which are guaranteed to converge after finitely many steps.
\end{theorem}
Theorem~\ref{th-guarded-effective} is a special case of the main result
of~\cite{BS-fmsd2012} (see also~\cite{kouzmin2004}) where it is
stated for arbitrary well-quasi-ordered sets (WQO's) and a generic notion
of ``effective regions''.
We recall that, by Higman's lemma,
$(\ConfS,\subword)$ is a \emph{well-quasi-ordered set}, i.e., a
quasi-ordered set ---$\subword$ is reflexive and transitive--- such that
every infinite sequence $c_0,c_1,c_2,\ldots$ contains an increasing
subsequence $c_i\subword c_j$ (with $i< j$).

\subsection{Stochastic games on lossy channel systems}

In the context of section~\ref{sec:plosses} and the stochastic arena
$\G_S$, we can reformulate the $\Pre$ operator used in
Section~\ref{sec-charac} as a regularity-preserving and effective operator.

When we only consider message losses, $\Pre[\delta](X) =
\Pre_S[\delta](\Cu X)$ (since $D_\env$ is
\emph{compatible} with the nondeterministic semantics).  If also
duplications are considered, then
$\Pre[\delta](X)=\Pre_S[\delta]\bigl(\invTdup(\Cu X)\bigr)$.
In order to deal uniformly with the two cases we shall let
$\Tdup$ be the identity relation when duplications are not considered.
By duality $\wPre[\delta](X)=\wPre_S[\delta](\Kd \winvTdup(X))$ and
the derived operators satisfy $\Pexists(X,Y) = \Pexists(\Kd
X,\Cu Y)$, $\Pforall(X,Y) = \Pforall(\Kd
X,\Cu Y)$ and $\PxA(X,Y) = \PxA(\Kd
X,\Cu Y)$.  Thus 
Theorem~\ref{th-correct-W} rewrites:
\begin{equation}
\label{eq-charac-WinA-lcs}
\WinA^{=1} \bigwedge_{i=1}^r \Box\Diamond R_i
\:=\:
\nu X.\bigcap_{i=1}^r H_i(X)
\:=\:
\nu X.\PxA\Bigl(\Kd \bigcap_{i=1}^r H_i(X), \ConfS\Bigr)
\:,
\end{equation}
with $H_i(X) \egdef\mu Z.X\cap \PxA\bigl(\Kd X,\Cu (R_i\cup Z)\bigr)$.
Observe how the closure properties of $\PxA$ let us easily rewrite $W'$
into a guarded term. The same technique does not apply to the simpler
term $W$ and this explains why we developed two characterizations of
the winning set in Section~\ref{sec-charac}. However, in the case
where $r=1$, the characterization with $W$ can be simplified in $W_1$
and Theorem~\ref{th-correct-W} yields the following \emph{guarded
  term} for stochastic B\"uchi games on lossy channel systems: \(
\WinA^{=1} \Box\Diamond R_1 = \nu X.\mu Z.\PxA\bigl(\Kd X,\Cu (R_1\cup
Z)\bigr). \)

Since $H_i(X)$ and $\WinA^{=1} \bigwedge_{i=1}^r \Box\Diamond R_i$
have guarded $L_\mu$ expressions, the following decidability result is
an immediate application of Theorem~\ref{th-guarded-effective} to
Eq.~\eqref{eq-charac-WinA-lcs}.
\begin{theorem}[Decidability of Generalized B\"uchi games with probability 1]
\label{th:decid}
In stochastic games on lossy channel system $S$ with regular arena
partition $\Conf_S = \ConfA\sqcup \ConfB$ and for regular goal regions
$R_1,\ldots,R_r$, the winning set $\WinA^{=1} \bigwedge_{i=1}^r
\Box\Diamond R_i$ is a regular region that can be computed uniformly
from $S$ and $R_1,\ldots,R_r$.
\end{theorem}
Furthermore, the winning strategies have simple finite
representations. One first computes the regular region $W$ ($=W'$).
Then for each rule $\delta\in\Delta$, and each $i=1,\ldots,r$, one
computes $V_i^\delta\egdef\ConfA\cap\Pre_S[\delta]\bigl(\Kd W\cap\Cu\bigl(R_i\cup H_i(W)\bigr)\bigr)$, these are again regular
regions. The strategy $\sigma_i$ for Alice is then ``when in
$V_i^\delta$, choose $\delta$'' and the strategy $\sigma_W$ is just a
combination of the $\sigma_i$'s using finite memory and testing when
we are in the $R_i$'s.

\paragraph{On complexity.}
Theorem~\ref{th-guarded-effective} does not only show that
$W=\WinA^{=1}\bigwedge_i\Box\Diamond R_i$ is computable from $S$ and
$R_1,\ldots,R_r$. It also shows that $W$ is obtained by computing  the sequence of
approximants $(W_{k})_{k\in\Nat}$ ---given by
$W_{0}=\ConfS$ and $W_{k+1} = \PxA(\Kd \HA(W_{k}),\ConfS)$--- until
the sequence stabilizes, which is guaranteed to eventually
occur. Furthermore, computing $\HA(W_{k})$, i.e., $\bigcap_{i=1}^r
H_i(W_{k})$, involves $r$ fixpoint computations that can use the same
technique: sequences of approximants guaranteed to converge in finite
time by Theorem~\ref{th-guarded-effective}.

There now exist generic upper bounds on the convergence time of such
sequences, see~\cite{SS-icalp11,SS-esslli2012}. In our case, they
entail that the above symbolic algorithm computing the regular region
$\WinA^{=1} (\bigwedge_{i=1}^r \Box\Diamond R_i)$ is in
$\bfF_{\omega^\omega}$, the first level in the Fast-Growing Complexity
hierarchy that is not multiply-recursive, hence has
``Hyper-Ackermannian'' complexity.

This bound is optimal: deciding whether $c \in \WinA^{=1}
(\bigwedge_{i=1}^r \Box\Diamond R_i)$ is $\bfF_{\omega^\omega}$-hard since this
generalizes reachability questions (on lossy channel systems)
that are $\bfF_{\omega^\omega}$-hard~\cite{CS-lics08,KS-fossacs2013}.
\begin{corollary}
Deciding whether $c \in \WinA^{=1} (\bigwedge_{i=1}^r \Box\Diamond R_i)$
for given $S$, $c$, $R_1$, \ldots, $R_r$ is $\bfF_{\omega^\omega}$-complete.
\end{corollary}



\section{Concluding remarks}
\label{sec-concl}

We gave a simple fixpoint characterization of winning sets and winning
strategies for 2-player stochastic games where a generalized B\"uchi
objective should be satisfied almost-surely. The characterization is
correct for any countable arena with a
finite attractor and satisfying the finite-choice assumption for Bob.

Such fixpoint characterizations lead to symbolic model-checking and
symbolic strategy-synthesizing algorithms for infinite-state systems
and programs. The main issue here is the finite-time convergence of
the fixpoint computations.  For well-quasi-ordered sets, one can use
generic results showing the finite-time convergence of so-called
``guarded'' fixpoint expressions as we demonstrated by showing the
decidability of generalized B\"uchi games on probabilistic lossy
channel systems, a well-quasi-ordered model that comes naturally
equipped with a finite attractor.

We believe that Theorem~\ref{th-guarded-effective} has more general applications
for games, stochastic or not, on well-quasi-ordered
infinite-state systems. We would like to mention quantitative objectives as
an interesting direction for future works (see~\cite{rabinovich2006,zielonka2010}).

\paragraph{Acknowledgements.}
We thank the anonymous referees who spotted a serious problem in the
previous version of this submission and made valuable suggestions that
let us 
improve the paper.



\ifthenelse{\boolean{eptcsstyle}}{ 
\bibliographystyle{eptcs}
\bibliography{patched}
}{
\bibliographystyle{plain}
\bibliography{../biblio}
}


\appendix

\section{Proof of Equation~\eqref{eq-mu-nu-contractive} page~\pageref{eq-mu-nu-contractive}}

Section~\ref{sec-charac} relies on the following Lemma for simplifying the
characterization of winning sets for simple B\"uchi objectives:

\begin{lemma}[Contractive $\bm{\nu}$-$\bm{\mu}$ fixpoint]
\label{lem-mu-nu-contractive}
For any binary (monotonic) operator $f$, $\nu X.\mu Y.X\cap f(X,Y) = \nu
X.\mu Y. f(X,Y)$.
\end{lemma}

This is a purely algebraic and lattice-theoretical result that is not
specific to stochastic games or channel systems. We include
its proof here for the sake of completeness.

We start with a simpler lemma: let $h$ be a unary (monotonic) operator.
\begin{lemma}
\label{lem-mu1}
Assume $U = \mu Y.h(Y)$ and $V\supseteq U$. Then $\mu Y.V\cap h(Y)=U$.
\end{lemma}
\begin{proof}
Write $W$ for $\mu Y.V\cap h(Y)$. Now $V\cap h(Y)\subseteq h(Y)$ entails
$\mu Y.V\cap h(Y)\subseteq \mu Y.h(Y)$, i.e., $W\subseteq U$, by
monotonicity.

For the other inclusion, we consider the approximants
$(U_\alpha)_{\alpha\in\Ord}$ of $U$ and show, by induction over $\alpha$,
that $U_\alpha\subseteq W$ for all $\alpha$, which is sufficient since
$U=\bigcup_\alpha U_\alpha$.

The base case $\alpha=0$ is clear since $U_0=\emptyset$. For the inductive
case $\alpha=\beta+1$, one has $U_{\alpha}\egdef h(U_\beta)$. From
$U_\beta\subseteq W$ (the ind.\  hyp.) we deduce $h(U_\beta)\subseteq h(W)$.
From $U_{\beta}\subseteq U$ and $h(U)=U$, we deduce $h(U_\beta)\subseteq
h(U)=U\subseteq V$. Thus $U_{\alpha}\subseteq V\cap h(W) = W$. Now for a
limit $U_\lambda$, we obtain $U_{\lambda}\subseteq W$ from
$U_{\lambda}=\bigcup_{\beta<\lambda}U_\beta$ and the ind.\  hyp.
\end{proof}

We may now prove Lemma~\ref{lem-mu-nu-contractive}. 
Write $g(X,Y)$ for $X\cap f(X,Y)$ and let $U\egdef\nu X.\mu Y.f(X,Y)$ and
$V\egdef\nu X.\mu Y.g(X,Y)$. From $g(X,Y)\subseteq f(X,Y)$ we derive
$V\subseteq U$ by monotonicity.

For the reverse inclusion, let $(V_\alpha)_{\alpha\in\Ord}$ be the
approximants of $V$. We claim that they satisfy the following inclusions
and equalities:
\begin{xalignat}{3}
\tag{P$_\alpha$, P$'_\alpha$, P$''_\alpha$}
\mu Y.f(V_\alpha,Y) &\subseteq V_\alpha
\:,
&
\mu Y.f(V_\alpha,Y) &= V_{\alpha+1}
\:,
&
U &\subseteq V_{\alpha}
\:,
\end{xalignat}
Note that (P$'_\alpha$) entails (P$_\alpha$) since $V_{\alpha_1}\subseteq
V_{\alpha_2}$ when $\alpha_1\geq \alpha_2$. Reciprocally (P$_\alpha$)
entails (P$'_\alpha$) since assuming (P$_\alpha$) and applying
Lemma~\ref{lem-mu1} on $h(Y)\egdef f(V_\alpha,Y)$ gives $\mu Y. f(V_\alpha,
Y)=\mu Y. V_\alpha\cap f(V_\alpha, Y)=\mu Y.g(V_\alpha,Y)$, which is the
definition of $V_{\alpha+1}$. Therefore it is sufficient to prove
(P$_\alpha$) and (P$''_\alpha$), which we do by induction over $\alpha$.

For the base case, (P$_0$) and (P$''_0$) are clear since $V_0\egdef \Conf$.

For the successor case $\alpha=\beta+1$, we start with $\mu
Y.f(V_\alpha,Y)\subseteq \mu Y.f(V_\beta,Y)$ ---by monotonicity, since
$V_\alpha\subseteq V_\beta$--- and combine with the
ind.\  hyp.\ (P$'_\beta$), i.e., $\mu Y.f(V_\beta,Y)=V_\alpha$, to obtain
(P$_\alpha$). For (P$''_\alpha$), we use the ind.\  hyp.\ $U\subseteq
V_\beta$ from which we deduce $\mu Y.f(U,Y)\subseteq\mu Y.f(V_\beta,Y)$,
i.e., $U\subseteq V_\alpha$, since $U=\mu Y.f(U,Y)$ by definition of $U$,
and $V_\alpha=\mu Y.f(V_\beta,Y)$ is the ind.\  hyp.\ (P$'_\beta$).

For the limit case $\alpha=\lambda$, one obtains (P$''_\lambda$) directly
from the ind.\  hyp.\  and the definition
$V_\lambda=\bigcap_{\beta<\lambda}V_\beta$. For (P$_\lambda$), we know $\mu
Y.f(V_\lambda,Y)\subseteq \mu Y.f(V_\beta,Y)$ for all $\beta<\lambda$ since
$V_\lambda\subseteq V_\beta$. Hence $\mu Y.f(V_\lambda,Y) \subseteq
\bigcap_{\beta<\lambda} \mu Y.f(V_\beta,Y) \subseteq\bigcap_{\beta<\lambda}
V_\beta$ (by ind.\  hyp.) $ = V_\lambda$.

Finally, since (P$''_\alpha$) holds for all $\alpha$ and since
$V=\bigcap_\alpha V_\alpha$, we deduce $U\subseteq V$.


\end{document}